\RequirePackage{amsmath,amssymb,amsfonts}
\documentclass[runningheads]{llncs}
\usepackage[T1]{fontenc}
\usepackage{graphicx}
\usepackage{multirow}%
\usepackage{mathtools}

\usepackage{mathrsfs}%
\usepackage[title]{appendix}%
\usepackage{xcolor}%
\usepackage{textcomp}%
\usepackage{manyfoot}%
\usepackage{booktabs}%
\usepackage{algorithm}%
\usepackage{algpseudocode}%
\usepackage{listings}%
\usepackage{multirow}

\usepackage{eurosym}
\usepackage{float}
\usepackage[tight,footnotesize]{subfigure}
\usepackage{rotating}
\usepackage{enumerate}
\usepackage{color,soul}
\usepackage[colorinlistoftodos]{todonotes}
\usepackage{lipsum}
\usepackage{cases}
\usepackage{comment}
\usepackage{resizegather}
\usepackage{xurl}
\usepackage{tablefootnote}
\usepackage{array}
\newcolumntype{C}[1]{>{\centering\arraybackslash$}p{#1}<{$}}

\usepackage{threeparttable}

\usepackage{xpatch}
\usepackage{nicematrix}

\newcommand{\bs}{\boldsymbol{\mathrm{s}}}
\newcommand{\br}{\boldsymbol{\mathrm{r}}}
\newcommand{\be}{\boldsymbol{\mathrm{e}}}

\newcommand{\bu}{\boldsymbol{\mathrm{u}}}
\newcommand{\bt}{\boldsymbol{\mathrm{t}}}
\newcommand{\bA}{\boldsymbol{\mathrm{A}}}
\newcommand{\buone}{\bA^{T}\br+\be_{1}}

\newcommand{\coloneqq}{:=}
\usepackage{lscape}

\DeclareMathOperator{\var}{Var}

 \DeclareMathOperator{\E}{E}

\DeclareMathSizes{10}{9}{7}{5}

\begin{document}
\emergencystretch 3em
\title{Semi-Compressed CRYSTALS-Kyber}

\author{{Shuiyin} {Liu}\inst{1}\orcidID{0000-0002-3762-8550} \and
{Amin} {Sakzad}\inst{2}\orcidID{0000-0003-4569-3384} }

\authorrunning{S. Liu and A. Sakzad}
\institute{ {Holmes
Institute}, {Melbourne, VIC 3000},  {Australia}
\\ \email{SLiu@holmes.edu.au}\\
\and
{Monash University}, {Melbourne, VIC 3800}, {Australia}\\
\email{Amin.Sakzad@monash.edu}}

\maketitle              % typeset the header of the contribution
\begin{abstract}

In this paper, we investigate the communication overhead of the Kyber, which has recently been standardized by the National Institute of Standards and Technology (NIST). Given the same decryption failure rate (DFR) and security argument, we show it is feasible to reduce the communication overhead of the Kyber by $54\%$. The improvement is based on two technologies: ciphertext quantization and plaintext encoding.  First, we prove that the Lloyd-Max quantization is optimal to minimize the decryption decoding noise. The original Kyber compression function is not optimal. Second, we propose an encoding scheme, which combines Pulse-Amplitude Modulation (PAM), Gray mapping, and a binary error correcting code. An explicit expression for the DFR is derived. The minimum possible communication overhead is also derived. Finally, we demonstrate that with the Lloyd-Max quantization, $8$-PAM, Gray mapping, and a shortened binary BCH$(768,638,13)$ code, the proposed scheme encapsulates $638$ bits (e.g., $2.5$ AES keys) in a single ciphertext.

%The abstract should briefly summarize the contents of the paper in 150--250 words.

\end{abstract}
\section{Introduction}
\vspace{-1mm} 

CRYSTALS-Kyber, the first post-quantum encryption algorithm selected by the National Institute of Standards and Technology (NIST), has attracted great attention from researchers and engineers \cite{NISTpqcdraft2023}. The security of Kyber derives from cautious parameterizations of the Module Learning with Errors (M-LWE) problem, which is widely believed to be post-quantum secure.  M-LWE based encryption approaches have a small failure probability during decryption, which is referred to as decryption failure rate (DFR). Observed decryption errors might leak the secret information to an adversary \cite{DFRattack2022}. Since the DFR is proportional to the value of the M-LWE module $q$, M-LWE based encryption approaches commonly choose a large $q$ to reduce the DFR. However, this choice leads to a large ciphertext size: for example, KYBER1024 converts a $32$-byte plaintext into a $1568$-byte ciphertext \cite{NISTpqcdraft2023}. 
Large ciphertexts are hard to communicate over the network and this is an obstacle to the adoption of Kyber in applications like Internet-of-Things (IoT).

The communication overhead is measured by the ciphertext expansion rate (CER). To reduce the CER, Kyber uses an almost uniform (lossy) compression function, denoted as, $\mathsf{Compress}_{q}(x,d)$,  to compress its ciphertexts \cite{Kyber2018}, where $x \in \mathbb{Z}_q$ and $d< \lceil \log_2(q)\rceil$. Since the function $\mathsf{Compress}_{q}(x,d)$ maps the fraction $x/q$ to the nearest fraction with denominator $2^d$, it can be viewed as an almost uniform quantization \cite{Azimi11}. From the decryption decoding perspective, the function $\mathsf{Compress}_{q}(x,d)$ introduces a $(\lceil \log_2(q)\rceil-d)$-bit quantization noise, resulting in a larger DFR. To reduce the DFR, Kyber uses two values $\{0, \left\lceil {q}/{2}\right\rfloor\}$ to represent one information bit, which is equivalent to an \emph{uncoded} binary Pulse Amplitude Modulation ($2$-PAM) \cite{Proakis}. The large Euclidean distance between $\{0, \left\lceil {q}/{2}\right\rfloor\}$ ensures that the DFR is sufficiently small. In summary, to reduce the CER and DFR, Kyber uses an almost uniform quantization and a $2$-PAM encoder.

From the plaintext encoding perspective, given the same DFR, it is possible to use error correcting codes to encode more information bits in a single ciphertext. The CER is thereby reduced. In \cite{liu2023lattice}, by replacing the $2$-PAM by the Leech lattice constellation, a $380$-bit plaintext (about $1.5$ AES keys) can be encrypted into a single ciphertext. The CER of Kyber is reduced by $32.6\%$. In \cite{LWEchannel2022}, a $471$-bit plaintext (about $1.8$ AES keys) is embedded in a single ciphertext, by using $5$-PAM and a Q-ary BCH code. The CER of Kyber is reduced by about $45.6\%$. In the broader literature of LWE, a variety of coding approaches have been proposed for different encryption schemes, including a concatenation of BCH and LDPC codes for NewHope-Simple \cite{NewhopeECC2018} and lattice codes for FrodoKEM \cite{FrodoCong2022}. Both approaches offer a CER reduction less than $15\%$. In practice, however, exchanging a slightly larger plaintext may have limited application, as only one AES key can be extracted from one ciphertext. An interesting question is if it is possible to encapsulate multiple 
$256$-bit secretes (e.g., AES keys, random seeds, initialization vectors) in a single ciphertext.

From the ciphertext compression perspective, Kyber uses an almost uniform quantization to compress the ciphertext \cite{Kyber2018}. It is known that for a continuous uniform input, the uniform quantization is optimal to minimize the mean squared error (MSE) \cite{Azimi11}. However, the ciphertext of Kyber is discrete and the quantization codebook size does not divide the module $q$ (i.e., a prime number). Kyber compression function is not necessarily optimal to decrease the MSE. Therefore, it is possible to reduce both the DFR and CER of Kyber by simply using a better quantization. In the literature, a few attempts have been made to develop better quantisation, including $\mathsf{D_4}$ lattice quantizer for NewHope \cite{NewHope2016} and $\mathsf{E_8}$ lattice quantizer for M-LWE \cite{MLWEE82021}. However, the CERs of \cite{NewHope2016}\cite{MLWEE82021} are larger than that of Kyber, leaving an important question open: whether an optimal ciphertext quantization can help reducing both the DFR and CER of Kyber.

From the security perspective, compressing the ciphertext is equivalent to adding an extra noise term to the ciphertext. It does not affect the security level of M-LWE based approaches. Although the error correcting code has no influence on the distribution of the ciphertext \cite{LWEchannel2022}, the implementation of error correcting decoding might be vulnerable to side-channel attacks. In \cite{ECCTimingAttack2019}, the authors point out that
error correcting decoding usually recovers valid codewords faster than the codewords that contain errors. This time information can be used to distinguish between valid ciphertexts and failing ciphertexts. However, this attack can be thwarted using a constant-time decoder. For example, a constant-time BCH decoder was proposed in \cite{constantBCH2020}. A constant-time Leech lattice decoder was given in \cite{constanttimeLeech2016}. With an optimal quantization and a constant-time encoding scheme, we expect that the CER can be reduced significantly with no security implications.

In this paper, we aim at reducing the CER of Kyber by more than $50\%$. First, we prove that the optimal quantization for M-LWE based approaches is the Lloyd-Max quantization \cite{Lloy1982}, which minimizes the MSE. Second, we present a variant of Kyber, where only the first part of ciphertext is quantized. This scheme is referred to as Semi-Compressed CRYSTALS-Kyber (SC-Kyber). This design is based on the fact that compressing the second part of ciphertext has little impact on reducing the CER, but at the cost of adding a $7$ or $8$-bit quantization noise to the decoding. Third, we show that the decryption decoding in the SC-Kyber is equivalent to the detection in the Additive white Gaussian noise (AWGN) channel. This result allows us to derive the maximum possible plaintext size, or equivalently, the minimum possible CER. Finally, we propose an encoding scheme for SC-Kyber. An explicit expression for the DFR is derived for such an encoding. We finally demonstrate that with the Lloyd-Max quantization, an $8$-PAM constellation, Gray mapping, and a shortened binary BCH$(768,638,13)$ code, the proposed scheme encrypts $638$ bits (e.g., $2.5$ AES keys) into a single ciphertext, given the same DFR and security level as in KYBER1024.  We summarize our results in Table \ref{sum_contribution}.

\vspace{-5mm}
\begin{table}[th]
\centering
\caption{Comparison of Variants of KYBER1024}
\label{sum_contribution}\centering
\vspace{-3mm} 
\begin{threeparttable}[b]
\begin{tabular}{|c|c|c|c|c|}
\hline
& Plaintext Size & CER  &  DFR & Source \\ \hline
KYBER1024-Uncoded & $256$ ($1$ AES Key) & $49$ & $2^{-174}$ &\cite{NISTpqcdraft2023}\cite{Kyber2021}\\ \hline
KYBER1024-Leech & $380$  ($1.5$ AES Keys) & $33$ & $2^{-226}$ & \cite{liu2023lattice}\\ \hline 
KYBER1024-Q-BCH & $471$  ($1.8$ AES Keys) & $26.6$ & $2^{-174}$ &  \cite{LWEchannel2022}\\ \hline
SC-KYBER1024-B-BCH & $638$  ($2.5$ AES Keys) & $22.5$ & $2^{-174}$ &  This work\\ \hline
\end{tabular}
%\vspace{-1mm}
\end{threeparttable}
\end{table}
\vspace{-8mm}

\section{Preliminaries}
\vspace{-2mm}
\subsection{Notation}

\emph{Rings:} Let $R_{q}$ denote the polynomial ring $\mathbb{Z}_{q}[X]/(X^{n}+1)$, where $n =256$ and $q=3329$ in this work. Elements in $R_{q}$ will be denoted with regular font letters, while a vector of the coefficients in $R_{q}$ is represented by bold lower-case letters. Matrices and vectors will be denoted with bold upper-case and lower-case letters, respectively. The transpose of a vector $\mathbf{v}$ or a matrix $\mathbf{A}$ is denoted by $\mathbf{v}^T$ or $\mathbf{A}^T$, respectively.  The default vectors are column vectors.

\emph{Sampling and Distribution:} For a set $\mathcal{S}$, we write $s \leftarrow \mathcal{S}$ to mean that $s$
is chosen uniformly at random from $\mathcal{S}$. If $\mathcal{S}$ is a probability
distribution, then this denotes that $s$ is chosen according to
the distribution $\mathcal{S}$. For a polynomial $f(x) \in R_q$ or a vector of such polynomials, this notation is extended coefficient-wise. We use $\var(\mathcal{S})$ to represent the variance of the distribution $\mathcal{S}$. Let $x$ be a bit string and $S$ be a distribution taking $x$ as the input, then $y\sim S\coloneqq\mathsf{Sam}\left(x\right)$ represents that the output $y$ generated by distribution $S$ and input $x$ can be extended to any desired length. We denote $\beta_{\eta}=B(2\eta
,0.5)-\eta $ as the central binomial distribution over $\mathbb{Z}$. The uniform distribution is represented as $\mathcal{U}$. The Cartesian product of two sets $A$ and $B$ is denoted by $A \times B$. We write  $A \times A = A^2$.

\emph{Kyber Compress and Decompress functions:} Let $x\in\mathbb{R}$ be a real number, then $\left\lceil x\right\rfloor $ means rounding to the closet integer with ties rounded up. The operations $\left\lfloor x\right\rfloor $ and $\left\lceil x \right\rceil$ mean rounding $x$ down and up, respectively. Let $x \in \mathbb{Z}_{q}$ and $d \in \mathbb{Z}$ be such that $2^d<q$. Kyber compression and decompression functions are \cite{NISTpqcdraft2023}:
{\setlength{\abovedisplayskip}{3pt} \setlength{\belowdisplayskip}{3pt}
\begin{align}
\mathsf{Compress}_{q}(x,d)&=\lceil (2^{d}/q)\cdot x\rfloor \mod 2^{d},\nonumber\\
\mathsf{Decompress}_{q}(x,d)&=\lceil (q/2^{d})\cdot x\rfloor.\label{ComDecom}
\end{align}}

\emph{Decryption Failure Rate (DFR) and Ciphertext Expansion Rate (CER):}
We let $\text{DFR} =\delta := \Pr (\hat{m} \neq m)$, where $m$ is a shared secret. It is desirable to have a small $\delta$, in order to be safe against decryption failure attacks \cite{DFRAttack2019}. In this work, the communication cost refers to the ciphertext expansion rate (CER), i.e., the ratio of the ciphertext size to the plaintext size.

\vspace{-3mm}
\subsection{Kyber Key Encapsulation Mechanism (KEM)}
Each message $m \in \{0,1\}^{n}$ can be
viewed as a polynomial in $R$ with coefficients in $\{0,1\}$. We recall Kyber.CPA = (KeyGen; Enc; Dec) \cite{Kyber2021} as described in Algorithms \ref{alg:kyber_keygen} to \ref{alg:kyber_dec}. The values of $\delta$, CER, and $(q, k, \eta_1,\eta_2, d_u, d_v)$ are given in Table \ref{Kyber_Par}. Note that the parameters $(q, k, \eta_1,\eta_2)$ determine the security level of Kyber, while the parameters $(d_u, d_v)$ describe the ciphertext compression rate.

\vspace{-5mm}
\begin{algorithm}[H]
\caption{$\mathsf{Kyber.CPA.KeyGen()}$: key generation}
\label{alg:kyber_keygen}
\begin{algorithmic}[1]

    \State
    $\rho,\sigma\leftarrow\left\{ 0,1\right\} ^{256}$

    \State
    $\bA\sim R_{q}^{k\times k}\coloneqq\mathsf{Sam}(\rho)$

    \State
    $(\bs,\be)\sim\beta_{\eta_1}^{k}\times\beta_{\eta_1}^{k}\coloneqq\mathsf{Sam}(\sigma)$

    \State
    $\bt\coloneqq\boldsymbol{\mathrm{As+e}}$\label{line:t}

    \State \Return $\left(pk\coloneqq(\boldsymbol{\mathrm{t}},\rho),sk\coloneqq\bs\right)$  

\end{algorithmic}
\end{algorithm}

\vspace{-12mm}

\begin{algorithm}[H]
\caption{$\mathsf{Kyber.CPA.Enc}$ $(pk=(\boldsymbol{\mathrm{t}},\rho),m\in\{0,1\}^{n})$}
\label{alg:kyber_enc}
\begin{algorithmic}[1]

	\State
	$r \leftarrow \{0,1\}^{256}$

	\State
	$\boldsymbol{\mathrm{A}}\sim R_{q}^{k\times k}\coloneqq\mathsf{Sam}(\rho)$
	
	\State  $(\boldsymbol{\mathrm{r}},\boldsymbol{\mathrm{e}_{1}},e_{2})\sim\beta_{\eta_1}^{k}\times\beta_{\eta_2}^{k}\times\beta_{\eta_2}\coloneqq\mathsf{Sam}(r)$
	
	\State  $\boldsymbol{\mathrm{u}}\coloneqq\mathsf{Compress}_{q}(\buone,d_{u})$\label{line:u}
	
	\State  $v\coloneqq\mathsf{Compress}_{q}(\boldsymbol{\mathrm{t}}^{T}\boldsymbol{\mathrm{r}}+e_2+\left\lceil {q}/{2}\right\rfloor \cdot m,d_{v})$\label{line:v}
	
	\State \Return $c\coloneqq(\boldsymbol{\mathrm{u}},v)$

\end{algorithmic}
\end{algorithm}

\vspace{-12mm}

\begin{algorithm}[H]
\caption{${\mathsf{Kyber.CPA.Dec}}\ensuremath{(sk=\bs,c=(\bu,v))}$}
\label{alg:kyber_dec}
\begin{algorithmic}[1]

    \State
    $\bu\coloneqq\mathsf{Decompress}_{q}(\bu,d_{u})$

    \State
    $v\coloneqq\mathsf{Decompress}_{q}(v,d_{v})$

    \State \Return $\mathsf{Compress}_{q}(v-\bs^{T}\bu,1)$

\end{algorithmic}
\end{algorithm}

\vspace{-10mm}

\begin{table}[ht]
\caption{Parameters of Kyber \cite{Kyber2021}\cite{NISTpqcdraft2023}}
\label{Kyber_Par}\centering
\vspace{-3mm}
\begin{tabular}{|c|c|c|c|c|c|c|c|c|c|}
\hline
& $k$ & $q$ & $\eta_{1}$ & $\eta_{2}$ & $d_{u}$ & $d_{v}$ & $\delta$ & CER & Information Bits \\ \hline
KYBER512 &  $2$ & $3329$ & $3$ & $2$ & $10$ & $4$ & $2^{-139}$ & $24$ & $256$ ($1$ AES key)\\ \hline
KYBER768 & $3$ & $3329$ & $2$ & $2$ & $10$ & $4$ & $2^{-164}$ &$34$ & $256$ ($1$ AES key)\\ \hline
KYBER1024 & $4$ & $3329$ & $2$ & $2$ & $11$ & $5$ & $2^{-174}$ &$49$ & $256$ ($1$ AES key) \\ \hline
\end{tabular}
\vspace{-2mm}
\end{table}

\subsection{Kyber Decryption Decoding}

Let $n_{e}$ be the decryption decoding noise in
Kyber. According to \cite{Kyber2021}, we have%
\begin{align}
n_{e}&=v-\mathbf{s}^{T}\mathbf{u} -\left\lceil q/2\right\rfloor
\cdot m  \notag \\ 
&= \mathbf{e}^{T}\mathbf{r}+e_{2}+c_{v}-\mathbf{s}^{T}\left( \mathbf{e}%
_{1}+\mathbf{c}_{u}\right),  \label{Ne}
\end{align}
where $c_{v}$ and $\mathbf{c}_{u}$ are rounding noises generated due to the compression operation. We let $\psi_{d_v}$ and $\psi_{d_u}^{k}$ be their respective distribution. The elements in $c_v$ or $\mathbf{c}_u$ are assumed to be i.i.d. and independent of other terms in (\ref{Ne}). The distribution $\psi _{d}$ is almost uniform over the integers in $[-\lceil q/2^{d+1}\rfloor ,\lceil q/2^{d+1}\rfloor]$ \cite{Kyber2018}. 

Kyber decoding problem can thus be formulated as%
\begin{equation}
y=v-\mathbf{s}^{T}\mathbf{u}=\left\lceil q/2\right\rfloor \cdot m+n_{e} \text{,}  \label{decoding_mode}
\end{equation}
i.e., given the observation $y\in R_{q}$, recover the value of $m$.

The encoding scheme in (\ref{decoding_mode}) can be viewed as a \emph{uncoded} $2$-PAM \cite{Proakis}, which can be easily generalized to the coded case  \cite{liu2023lattice}\cite{LWEchannel2022}:
\begin{equation}
y=\left\lceil q/p\right\rfloor \cdot \mathsf{ENC}(m)+n_{e} \text{,}  \label{decoding_mode_coded}
\end{equation}
where $p \in [2, \sqrt{2q})$ is an integer and the function $\mathsf{ENC}(\cdot)$ represents an encoder. For example, \cite{liu2023lattice} uses a lattice encoder, while \cite{LWEchannel2022} uses a Q-ary BCH encoder.

The distribution of $n_e$ can be evaluated numerically \cite{Kyber2021}. To gain more insight, \cite{liu2023lattice} shows that $n_e$ is well-approximated by the sum of multivariate normal and a discrete uniform random vector, i.e.,
\begin{equation}
n_{e}\leftarrow \mathcal{N}(0,\sigma _{G}^{2}I_{n}) +\mathcal{U}%
(-\lceil q/2^{d_{v}+1}\rfloor ,\lceil
q/2^{d_{v}+1}\rfloor ),  \label{D_ne}
\end{equation}%
where $\sigma _{G}^{2}=kn\eta _{{1}}^2/4+ kn\eta_{_1}/2 \cdot (\eta_2/2+\var(\psi_{d_u}))+\eta_2/2$.

Equation (\ref{D_ne}) explains how the compression parameters $(d_u,d_v)$ affect the decoding noise. An open question is if the Kyber compression function in (\ref{ComDecom}) is optimal to minimize the DFR. We will answer this question in the next section.

\vspace{-3mm}

\section{Optimal Quantization for Kyber}
\vspace{-1mm}
We consider the compression, or more generally, the quantization in Kyber. Without loss of generality, we define a general quantization function 
\begin{equation}
\hat{\mathbf{x}}=Q_{L}(\mathbf{x}, \mathcal{C}_L, T_L), \label{Quantizer}
\end{equation}
where a random vector $\mathbf{x} \in \mathbb{Z}_{q}^{n}$ is quantized to a quantizer $\hat{\mathbf{x}} \in \mathcal{C}_L$, according to the quantization codebook $\mathcal{C}_L \in \mathbb{R}^n$ and the corresponding decision regions $T_L \subset \mathbb{R}^n$. The number $L$ represents the number of quantizers in $\mathcal{C}_L$. 

We define the unique index of $\hat{\mathbf{x}}$ in $\mathcal{C}_L$ as
$\mathsf{Index}_L(\hat{\mathbf{x}})$, i.e.,
\begin{equation}
\mathcal{C}_L(\mathsf{Index}_L(\hat{\mathbf{x}})) = \hat{\mathbf{x}}. \label{ind_c}
\end{equation}
When storing/transmitting $\hat{\mathbf{x}}$, we only need to save/send $\mathsf{Index}_L(\hat{\mathbf{x}})$. Let $\mathbf{e}_{L} = \mathbf{x}-\hat{\mathbf{x}}$ be the quantization error vector. The mean squared error (MSE), denoted as
\begin{equation}
\mathsf{MSE}(\mathbf{e}_{L}) = \mathsf{E}(\|\mathbf{x}-\hat{\mathbf{x}}\|^2),
\end{equation}
is almost invariably used in this text to measure distortion. The optimal quantization should achieve minimum MSE (MMSE):
\begin{equation}
(\mathcal{C}_L, T_L) = \arg \min_{\mathcal{C}'_L \in \mathbb{R}^n, T_L'\subset \mathbb{R}^n }\mathsf{E}(\|\mathbf{x}-Q_{L}(\mathbf{x}, \mathcal{C}'_L, T_L')\|^2). \label{MMSE}
\end{equation}
For simplicity of notation, we define the MMSE quantization as
\begin{equation}
\hat{\mathbf{x}}=Q_{\mathsf{MMSE},L}(\mathbf{x}).
\end{equation}

\begin{comment}
\begin{equation}
(\mathcal{C}_L, T) = \arg \min_{Q_L(\cdot)}\mathsf{E}(\|\mathbf{x}-Q_{L}(\mathbf{x}, \mathcal{C}'_L, T')\|^2). \label{MMSE}
\end{equation}
\end{comment}

In this section, we will show Kyber compression function in (\ref{ComDecom}) is not optimal to decrease the MSE. We will prove that the Lloyd-Max quantization \cite{Lloy1982} produces the optimal quantization codebook and decision regions for the Kyber.

\vspace{-2mm}
\subsection{The Lloyd-Max Quantization (LMQ)}

The Lloyd-Max quantization \cite{Lloy1982} is a scalar quantization, which minimizes MSE distortion with a fixed number of regions. We consider the situation with $L$ quantizers $\mathcal{C}_L=\{\hat{x}_1,\ldots, \hat{x}_L\}$. Let the corresponding quantization intervals be
\begin{equation}
T_L=\{ (\beta_{i-1},\beta_{i}), i=1,\ldots, L \},
\end{equation}%
where $\beta_{0}=-\infty$ and $\beta_{L}=\infty$. Given a bounded random variable $x$, let $x_{\min}= \min(x)$, $x_{\max}= \max(x)$, and $\Pr(x)$ be the probability mass function (PMF) of $x$. The Lloyd-Max algorithm takes inputs $(x, \Pr(x))$ and produces $(\mathcal{C}_L, T_L)$.

\begin{enumerate}
    \item Choose an arbitrary initial set of $L$ representation points $\{\hat{x}_i\}_{i=1}^{L}$ in ascending order. A common choice is $\hat{x}_i = x_{\min}+ i \cdot (x_{\max}-x_{\min})/L, i=1,\ldots, L$.
    \item For each $i$; $1 \leq i \leq L-1$, set ${\beta}_i =  (\hat{x}_i + \hat{x}_{i+1})/2$.
    \item For each $i$; $1 \leq i \leq L$, set $\hat{x}_i$ equal to the conditional mean of $x$ given $x \in (\beta_{i-1},\beta_{i}]$. See line 8 of Algorithm~\ref{alg:LloydMAX_Pseudocode}.
    \item Repeat Steps 2 and 3 until further improvement in MSE is negligible.  
\end{enumerate}
The MSE decreases (or remains the same) for each iteration. As shown in Step 2, each decision threshold ${\beta}_i$ is in the middle of two consecutive quantizers. Therefore, the input $x$ will be quantized to the nearest quantizer, i.e.,
\begin{equation}
\hat{x}= Q_{\mathsf{MMSE},L}({x}) = \arg \min_{\hat{x}' \in \mathcal{C}_L}\|{x}-\hat{x}'\|. \label{MMSE_MAP}
\end{equation}
The pseudocode is given in Algorithm \ref{alg:LloydMAX_Pseudocode} (the discrete version of LMQ in \cite{Scheunders1996}). 

\vspace{-7mm}

\begin{algorithm}[H]
\caption{$[\mathcal{C}_L, T_L] = \mathsf{LloydMax}(x,\Pr(x))$}
\label{alg:LloydMAX_Pseudocode}
\begin{algorithmic}[1]

	\State
	$\hat{x}_i \coloneqq x_{\min} + i \cdot (x_{\max}-x_{\min})/L, i=1,\ldots, L$  \Comment{initial set of $\mathcal{C}_L$ }

 \State {$\hat{\beta}_{0}\coloneqq -\infty$ and $\hat{\beta}_{L}\coloneqq \infty$}

  \Repeat

   \For{$i \gets 1$ to $L-1$}                    
        \State {${\beta}_i \coloneqq (\hat{x}_i + \hat{x}_{i+1})/2$}
  \EndFor   \Comment{update $T_L$}

\For{$i \gets 1$ to $L$}                    
        \State {$\hat{x}_i \coloneqq \dfrac{\sum_{x \in (\beta_{i-1},\beta_{i}]}\Pr(x)x}{{\sum_{x \in (\beta_{i-1},\beta_{i}] }\Pr(x)}}$}    \Comment{the conditional mean}
  \EndFor    
  
  \Until{$T$ does not change}
	
\State \Return $\mathcal{C}_L, T_L$

\end{algorithmic}
\end{algorithm}

\vspace{-5mm}

Since the M-LWE samples are assumed to be uniformly distributed, the following lemma gives the optimal quantization for the Kyber.

\begin{lemma}[Global Minimum]\label{lem:gm} Let $\mathbf{x}  = [x_1,\ldots, x_n]^T  \leftarrow \mathbb{Z}_{q}^{n}$ and 
$ L_i \in \mathbb{Z}_{q}  \backslash \{0\}$ for $i=1, \ldots, n$. If
$[\mathcal{C}_{L_i}, T_{L_i}] = \mathsf{LloydMax}(x_i,\Pr(x_i))$
for $i=1, \ldots, n$, then 
$C_L = C_{L_1} \times \cdots \times C_{L_n}$ and $T_L=T_{L_1} \times \cdots \times T_{L_n}$ is the global solution to the following MMSE quantization problem
\begin{equation}
(\mathcal{C}_L, T_L) = \arg \min_{\mathcal{C}'_L \in \mathbb{R}^n, T_L'\subset \mathbb{R}^n }\mathsf{E}(\|\mathbf{x}-Q_{L}(\mathbf{x}, \mathcal{C}'_L, T_L')\|^2). \label{MMSE_t2}
\end{equation}
\end{lemma}

\begin{proof}
For $x_i$, $i=1, \ldots, n$, we define $L_i$ quantizers $C_{L_i} = \{\hat{x}_{i,1},\ldots, \hat{x}_{i,L_i}\}$ and the corresponding quantization intervals
\begin{equation}
T_{L_i}=\{(\beta_{i,j-1},\beta_{i,j}), j=1,\ldots, L_i\},
\end{equation}%
where $\beta_{i,0}=-\infty$ and $\beta_{i,L_i}=\infty$.

Since $\mathbf{x} \leftarrow \mathbb{Z}_{q}^{n}$, the MSE can be written as follows:
\begin{equation}
\sum_{\mathbf{x} \in \mathbb{Z}_{q}^{n}}\Pr(\mathbf{x})\|\mathbf{x}-Q_{L}(\mathbf{x}, \mathcal{C}_L, T_L)\|^2 =\Pr(\mathbf{x})\sum_{i=1}^{n} \left( D_i \prod_{\ell \neq i}L_{\ell} \right), \label{D_MSE_G}
\end{equation}
where
\vspace{-5mm}
\begin{equation}
D_i= \sum_{j=1}^{L_i}\sum_{x_i \in (\beta_{i,j-1},\beta_{i,j}] } \|{x_i}-\hat{x}_{i,j}\|^2. 
\end{equation}

To minimize the MSE, we take derivatives of Equation (\ref{D_MSE_G}) with respect to $(\hat{x}_{i,j}, \beta_{i,j})$ and set them equal to zero, leading to the following conditions  for the optimum quantizers $\hat{x}_{i,j}$ and quantization interval boundaries ${\beta}_{i,j}$ \cite{Scheunders1996}\cite{Clapp2006}:
\begin{align}
{\beta}_{i,j} & =  \dfrac{\hat{x}_{i,j}+\hat{x}_{i,j+1}}{2} \notag \\ 
\hat{x}_{i,j} &= \dfrac{\sum_{x_i \in (\beta_{i,j-1},\beta_{i,j}] }x_i}{|(\beta_{i,j-1},\beta_{i,j}] \cap \mathbb{Z}|}. \label{g_cond}
\end{align}
We observe that the optimal $(\hat{x}_{i,j},{\beta}_{i,j})$ only depends on $x_i$. Therefore, any local minimum is also a global minimum. 

A way to solve Equation (\ref{g_cond}) is to first generate an initial set  $\{\hat{x}_{i,1},\ldots, \hat{x}_{i,L_i}\}$, then apply Equation (\ref{g_cond}) alternately until convergence is obtained. This iteration is well known as the Lloyd-Max quantization \cite{Lloy1982}. 
\end{proof}

In Table \ref{Lloyd_MSE}, we compare the values of MSE for Lloyd-Max quantization and Kyber compression. We quantize $x \leftarrow \mathbb{Z}_{q}$ with $L=2^{d}$ quantization levels. In the Kyber, $d=10$ or $11$ is used  to compress the first part of ciphertext, i.e., $\mathbf{u}$.  For $d=11$, we observe that Lloyd-Max quantization has a much smaller MSE. 
\begin{table}[ht]
\caption{$\mathsf{MSE}({e}_{L})$: Kyber Compression vs. Lloyd-Max Quantization}
\label{Lloyd_MSE}\centering
\vspace{-3mm} 
\begin{tabular}{|c|c|c|c|}
\hline
$d$ & $11$ & $10$  \\ \hline
$\text{Kyber Compression}$  &$0.38$ \cite{liu2023lattice} & $0.92$ \cite{Kyber2021} \\ \hline
$\text{Lloyd-Max}$  &$0.1924$ & $0.8468$ \\ \hline
\end{tabular}
\vspace{-5mm}
\end{table}

\begin{remark}
With a continuous uniform input, the Lloyd-Max algorithm returns a set of evenly-spaced intervals that span the range of the input \cite{Azimi11}. With a discrete uniform input $x \leftarrow \mathbb{Z}_{q}$, however, the optimal intervals are not necessarily evenly-spaced. Considering the fact that $q$ is prime, the intervals returned by Lloyd-Max algorithm will have different sizes. The quantization codebook is not obvious.
\end{remark}

\begin{example}
With $x \leftarrow \mathbb{Z}_q$ and $L=2^{d}$, we present the distribution of $e_L = {x}-\hat{{x}}$ in Tables \ref{PMF_com_d11}  and \ref{PMF_com_d10}. We observe that for the Lloyd-Max quantization, the $e_L$ is not uniformly distributed. It confirms that with the discrete uniform input, a uniform quantization is not necessarily optimal to reduce the MSE.
\end{example}

\vspace{-8mm}

\begin{table}[ht]
\caption{PMF of $e_L$ with $L= 2^{11}$}
\label{PMF_com_d11}\centering
\vspace{-6mm} 
\begin{center}
\begin{tabular}{|c|C{1.4cm}|C{1.3cm}|C{1.3cm}|}
\hline
\multicolumn{4}{|c|}{Kyber Compression} \\ \hline
$e_L$ & $-1$ & $0$  & $1$   \\ \hline
$\Pr(e_L)$  &  $0.1916$ & $0.6138$  & $0.1946$ \\ \hline
\end{tabular}

\begin{tabular}{|c|C{1.4cm}|C{1.3cm}|C{1.3cm}|}
\hline
\multicolumn{4}{|c|}{Lloyd-Max} \\
 \hline
$e_L$ & $-0.5$ & $0$  & $0.5$ \\ \hline
$\Pr(e_L)$  &  $0.3848$ & $0.2304$  & $0.3848$ \\ \hline
\end{tabular}
\end{center}
\vspace{-15mm}
\end{table}

%\vspace{-12mm}

\begin{table}[ht]
\caption{PMF of $e_L$ with $L= 2^{10}$}
\label{PMF_com_d10}\centering
\vspace{-6mm} 
\begin{center}
\begin{tabular}{|c|C{1.4cm}|C{1.3cm}|C{1.3cm}|C{1.3cm}|C{1.2cm}|}
\hline
\multicolumn{6}{|c|}{Kyber Compression} \\ \hline
$e_L$ & $-2$ & $-1$  & $0$ & $1$ & $2$  \\ \hline
$\Pr(e_L)$  &  $0.0390$ & $0.3061$  & $0.3068$ & $0.3088$ & $0.0393$  \\ \hline
\end{tabular}

\begin{tabular}{|c|c|c|c|c|c|c|c|}
\hline
\multicolumn{8}{|c|}{Lloyd-Max} \\
 \hline
$e_L$ & $-1.5$ & $-1$  & $-0.5$ & $0$  & $0.5$ & $1$ & $1.5$  \\ \hline
$\Pr(e_L)$  &  $0.0772$ & $0.2304$  & $0.0772$ & $0.2304$ & $0.0772$ & $0.2304$ & $0.0772$  \\ \hline
\end{tabular}
\end{center}
\vspace{-5mm}
\end{table}

\subsection{Kyber with the Lloyd-Max Quantization}

We replace the Kyber compression function in Algorithm \ref{alg:kyber_enc} by the Lloyd-Max quantization in Algorithm \ref{alg:LloydMAX_Pseudocode}. The quantized ciphertext $(\mathbf{u},v)$ is given by
\begin{equation}
 (\mathbf{u} = Q_{\mathsf{MMSE},L_u}(\buone), v=Q_{\mathsf{MMSE},L_v}(\boldsymbol{\mathrm{t}}^{T}\boldsymbol{\mathrm{r}}+e_2+\left\lceil {q}/{2}\right\rfloor \cdot m)).
\end{equation}
where $L_u=2^{knd_u}$ and $L_v=2^{nd_v}$. The corresponding quantization codebooks and decision regions are given by $(\mathcal{C}_{L_u} = \mathcal{C}_{2^{d_u}}^{kn} , T_{L_u}=T_{2^{d_u}}^{kn})$ and $(\mathcal{C}_{L_v}=\mathcal{C}_{2^{d_v}}^{n}, T_{L_v}=T_{2^{d_v}}^{n})$, respectively. They are obtained from Lemma \ref{lem:gm} and Algorithm \ref{alg:LloydMAX_Pseudocode}. 

When storing/transmitting a ciphertext, we only need to save/send the indices of the quantized ciphertext coefficients:
\begin{equation}
 (\mathsf{Index}_{L_u}(\mathbf{u}), \mathsf{Index}_{L_v}(v)).
\end{equation}
The size of the ciphertext remains the same as the original Kyber.  With the Lloyd-Max quantization, the updated Kyber encryption and decryption algorithms are described below.  The key generation function is the same as Algorithm \ref{alg:kyber_keygen}. Thus we omit details here.

\vspace{-5mm}

\begin{algorithm}[H]
\caption{$\mathsf{Kyber.LloydMax.CPA.Enc}$ $(pk=(\boldsymbol{\mathrm{t}},\rho),m\in\{0,1\}^{n})$}
\label{alg:kyber_enc_LM}
\begin{algorithmic}[1]

	\State
	$r \leftarrow \{0,1\}^{256}$

	\State
	$\boldsymbol{\mathrm{A}}\sim R_{q}^{k\times k}\coloneqq\mathsf{Sam}(\rho)$
	
	\State  $(\boldsymbol{\mathrm{r}},\boldsymbol{\mathrm{e}_{1}},e_{2})\sim\beta_{\eta_1}^{k}\times\beta_{\eta_2}^{k}\times\beta_{\eta_2}\coloneqq\mathsf{Sam}(r)$
	
	\State  $\boldsymbol{\mathrm{u}}\coloneqq Q_{\mathsf{MMSE},L_u}(\buone)$\label{line:u_LM}
	
	\State  $v\coloneqq Q_{\mathsf{MMSE},L_v}(\boldsymbol{\mathrm{t}}^{T}\boldsymbol{\mathrm{r}}+e_2+\left\lceil {q}/{2}\right\rfloor \cdot m)$\label{line:v_LM}
	
	\State \Return $c\coloneqq(\mathsf{Index}_{L_u}(\mathbf{u}),\mathsf{Index}_{L_v}(v))$

\end{algorithmic}
\end{algorithm}

\vspace{-12mm}

\begin{algorithm}[H]
\caption{${\mathsf{Kyber.LloydMax.CPA.Dec}}\ensuremath{(sk=\bs,c=(\mathsf{Index}_{L_u}(\mathbf{u}),\mathsf{Index}_{L_v}(v)))}$}
\label{alg:kyber_dec_LM}
\begin{algorithmic}[1]

    \State
    $\bu\coloneqq \mathcal{C}_{L_u}(\mathsf{Index}_{L_u}(\boldsymbol{\mathrm{u}}))$

    \State
    $v\coloneqq \mathcal{C}_{L_v}(\mathsf{Index}_{L_v}(v))$

    \State \Return $\mathsf{Compress}_{q}(v-\bs^{T}\bu,1)$

\end{algorithmic}
\end{algorithm}

\vspace{-5mm}

We then study the DFR of the optimized Kyber in Algorithms \ref{alg:kyber_enc_LM} and \ref{alg:kyber_dec_LM}. Let  $(\mathbf{e}_{L_u}$, ${e}_{L_v})$ be the quantization noise for $(\mathbf{u},v)$, respectively. The Kyber decoding noise $n_{e}$ in (\ref{Ne}) can be rewritten as
\begin{align}
n_{e} = \mathbf{e}^{T}\mathbf{r}+e_{2}+{e}_{L_v}-\mathbf{s}^{T}\left( \mathbf{e}%
_{1}+\mathbf{e}_{L_u}\right).  \label{Ne_q}
\end{align}
Following the same line of Kyber \cite{Kyber2021}, the elements in $\mathbf{e}_{L_u}$ or ${e}_{L_v}$ are assumed to be i.i.d. and independent of other terms in (\ref{Ne_q}). 

Without loss of generality, we use $e_{L}^{(1)} \in \mathcal{S}$ to present an element in $\mathbf{e}_{L}$ or ${e}_{L}$. Since  $\E\left({e}_{L}^{(1)}\right)=0$, we use a variant of \cite[Theorem 1]{liu2023lattice}:
\begin{theorem}\label{the:CLT_noise}
According to the Central Limit Theorem (CLT), the distribution of $n_e$ asymptotically approaches the sum of multivariate normal and quantization error ${e}_{L_v}$:
{\setlength{\abovedisplayskip}{3pt} \setlength{\belowdisplayskip}{3pt}
\begin{equation}
n_{e}\leftarrow \mathcal{N}(0,\sigma _{G}^{2}I_{n}) +{e}_{L_v}, 
\end{equation}%
}where $L_u=2^{knd_u}$, $L_v=2^{nd_v}$,and $\sigma _{G}^{2}=kn\eta _{{1}}^2/4+ kn\eta_{_1}/2 \cdot \left(\eta_{_2}/2 + \mathsf{MSE}\left({e}_{L_u}^{(1)}\right)\right)+\eta_2/2$. The values of $\mathsf{MSE}({e}_{L_u}^{(1)})$ are given in Table \ref{Lloyd_MSE}.
\end{theorem}

\begin{remark}
Theorem \ref{the:CLT_noise} shows how the choice of quantization affects the Kyber decoding noise. For a fixed $(q, k, \eta_1, \eta_2, d_u, d_v)$ in Table \ref{Kyber_Par}, a MMSE quantization can minimize the decoding noise variance $\var(n_e)$. The CLT assumption has been made in the Ring/Module-LWE literature. In \cite{TFHE2018},  the noise coefficients of TFHE scheme after bootstrapping are assumed to be independent Gaussians. The assumption is experimentally verified in \cite[Figure 10]{TFHE2018}. In \cite{CLTRLWE2022}, the noise coefficients of CKKS scheme \cite{cheon2017homomorphic} are assumed to be independent Gaussians. In \cite{LWEchannel2022}, the noise coefficients of Kyber are assumed to be independent. Based on the CLT assumption, an  upper bound on DFR is derived in \cite{liu2023lattice},
which is very close to the numerical bound in \cite{Kyber2021}. We will verify the CLT assumption below.
\end{remark}

Consider the noise term $Y= \mathbf{e}^{T}\mathbf{r}+e_{2}-\mathbf{s}^{T}\left( \mathbf{e}_{1}+\mathbf{e}_{L_u}\right)$. Let $Y=[Y_{1}, \ldots,Y_{n}]^T$ be the coefficients in $Y$. In Fig. \ref{K_CDF_Plot}, we compare the cumulative distribution functions (CDF) of $Y_{1}/\sigma _{G}$ (in blue) and the standard normal distribution (in red). We observe that the curves are almost indistinguishable. Experimental result confirms the distribution of $Y_{1}$ is Gaussian. This experimentally validates our independent assumption in Theorem \ref{the:CLT_noise}, since $Y_{i}$, for $i=1, \ldots, n$,  are uncorrelated and identically distributed random variables.

\begin{figure}[tbp]
\centering
\includegraphics[width=0.8\textwidth]{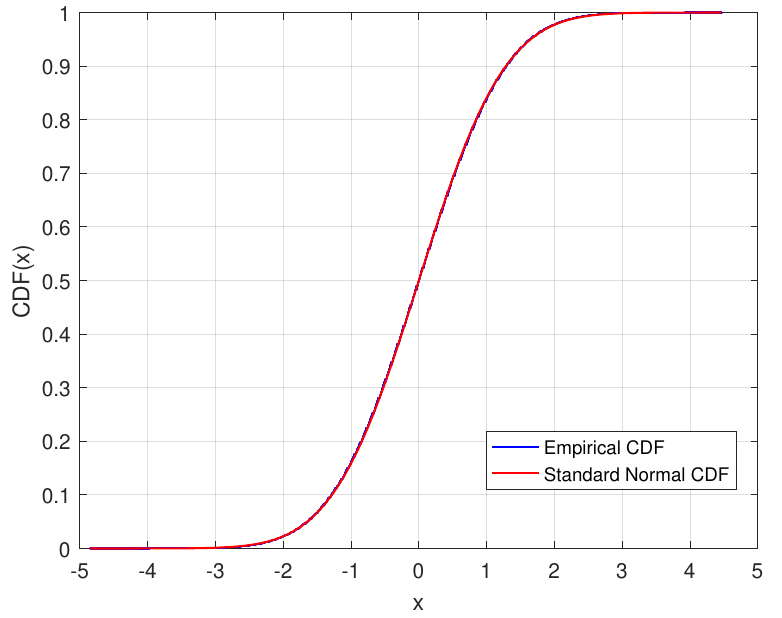} \vspace{-5mm} 
\caption{KYBER1024: Comparing the CDF of $Y_{1}/\sigma_{G}$
to the standard normal distribution with $1,000,000$ samples, experimental validation of CLT assumption}
\label{K_CDF_Plot}
\vspace{-5mm} 
\end{figure}

Using Theorem \ref{the:CLT_noise}, we derive the DFR as follows.

\begin{theorem}\label{the:LM_DFR}
With the Lloyd-Max quantization, the DFR of Kyber is given by
\begin{align}
\delta & = 1-\left(1-%
{\textstyle\sum\nolimits}_{e_{L_v}^{(1)} \in \mathcal{S}}
\Pr\left(e_{L_v}^{(1)}\right)Q_{1/2}\left( %
\Vert e_{L_v}^{(1)} \Vert /\sigma _{G},\left\lceil {q}/{4}\right\rfloor/{\sigma
_{G}}\right) \right)^{n}, \label{Kyber_Q_DFR}
\end{align}
where $\sigma_{G}$ is given in Theorem \ref{the:CLT_noise}, and $Q_{M}\left( a,b\right) $ is the generalized Marcum Q-function.  Both $e_{L_v}^{(1)}$ 
and $\Pr\left(e_{L_v}^{(1)}\right)$ can be found by Algorithm \ref{alg:LloydMAX_Pseudocode}.
\end{theorem}

\begin{proof}
Given Theorem \ref{the:CLT_noise}, we can write
$n_{e}=x+e_{L_v}$, where $x\leftarrow \mathcal{N}( 0,\sigma
_{G}^{2}I)$. Since the elements in $n_{e}$ are i.i.d., for simplicity, we consider one element in $n_{e}$, denoted as $n_{e}^{(1)}=x^{(1)}+e_{L_v}^{(1)}$. The DFR can be written as 
\begin{equation}
\delta = 1-\Pr(\Vert n_{e}^{(1)} \Vert \leq \lceil q/4\rfloor )^n. \label{p0}
\end{equation}
We have
\begin{equation}
\Pr \left( \Vert n_{e}^{(1)}\Vert \leq z\right)  ={\textstyle\sum\nolimits}_{e_{L_v}^{(1)} \in \mathcal{S}}\Pr \left(  \Vert x^{(1)}+e_{L_v}^{(1)} \Vert \leq z\left\vert e_{L_v}^{(1)}\right. \right) \Pr \left(e_{L_v}^{(1)}\right). \label{p1}
\end{equation}
Given $e_{L_v}^{(1)}$, $\Vert x^{(1)}+e_{L_v}^{(1)} \Vert$ follows non-central chi distribution, i.e.,
\begin{equation}
\Pr \left(  \Vert x^{(1)}+e_{L_v}^{(1)} \Vert \leq z\left\vert e_{L_v}^{(1)}\right. \right) = 1-Q_{1/2}\left( \Vert e_{L_v}^{(1)} \Vert/{\sigma
_{G}},{z}/{\sigma
_{G}}\right), \label{p2}
\end{equation}
where $Q_{M}\left( a,b\right) $ is the generalized Marcum Q-function. By substituting (\ref{p2}) and (\ref{p1}) for (\ref{p0}), we obtain (\ref{Kyber_Q_DFR}). 
\end{proof}

In Table \ref{Com_KC_MSE}, we compare the DFRs of Kyber with different quantizations. We observe that the DFR is noticeably decreased by using a MMSE quantization.
\vspace{-4mm}
\begin{table}[ht]
\caption{DFR: Kyber Compression vs. Lloyd-Max}
\label{Com_KC_MSE}\centering
\vspace{-3mm} 
\begin{tabular}{|c|c|c|}
\hline
  & Original Compression &  Lloyd-Max  \\ \hline
KYBER512    & $2^{-139}$ & $2^{-150}$  \\ \hline
KYBER768    & $2^{-164}$ & $2^{-177}$  \\ \hline
KYBER1024    & $2^{-174}$ & $2^{-196}$  \\ \hline
\end{tabular}
\vspace{-10mm}
\end{table}

\subsection{Complexity of Lloyd-Max Quantization}
\vspace{-1mm}
Kyber compression function in (\ref{ComDecom}) involves one multiplication and one division operations, which imply time complexity $O(\log_2(q)^2)$. Here we assume that the rounding and modulo operations have very low computational cost and are treated as constant-time operations. The complexity of the decompression function is the same as the compression function.

For the Lloyd-Max quantization, an input $x$ will be mapped to the nearest quantizer. Given a pre-stored codebook $\mathcal{C}_L$, we can use a binary search method to find the nearest quantizer, with time complexity $O(\log_2(L))$. Considering that the Lloyd-Max quantization does not require a reconstruction function, we conclude that the Lloyd-Max quantization takes less time than Kyber compression.

Note that the security levels of Kyber are computed independent of the compression method and the compression noise level \cite{Kyber2021}. Therefore, using Lloyd-Max quantization does not affect the security argument of Kyber.

\vspace{-2mm}

\section{Semi-Compressed Kyber}
\subsection{The Design}
We first revisit the ciphertext compression strategy in Kyber (Algorithm \ref{alg:kyber_enc}):
\begin{itemize}
    \item Step 4: compress the first part of ciphertext, $\mathbf{u}$, to $knd_u$ bits.
    \item Step 5: compress the second part of ciphertext, $\mathbf{v}$, to $nd_v$ bits.
\end{itemize}
The CER of Kybe can be computed as
\begin{equation}
\text{CER}=\dfrac{knd_u+nd_v}{K}, 
\end{equation}%
where $K$ represents the number of information bits. In the original Kyber, $K$ is set to be $256$ bits. Considering the values of $(n,k, d_u, d_v)$ in Table \ref{Kyber_Par}, we notice that compressing $\mathbf{v}$ has little impact on reducing CER, but at the cost of adding a $7$ or $8$-bit quantization noise to the decoding.

\begin{definition}[SC-Kyber] We consider a variant of Kyber as follows.
\begin{equation}
 (Q_{\mathsf{MMSE},L_u}(\buone),\boldsymbol{\mathrm{t}}^{T}\boldsymbol{\mathrm{r}}+e_2+\left\lceil {q}/{p}\right\rfloor \cdot \mathsf{ENC}(m)),
\end{equation}
where $p \geq 2$, $m \leftarrow \{0,1\}^K$, and $\mathsf{ENC}(m) :  \{0,1\}^K \mapsto \mathbb{Z}_p^{n}$.
Since only the first part of ciphertext is quantized, 
this scheme is referred to as \emph{Semi-Compressed Kyber} (SC-Kyber).  The CER reduces to
\begin{equation}
\text{CER}=\dfrac{knd_u+12n}{K}. 
\end{equation}
\end{definition}

We expect that the reduced decryption decoding noise enables encoding a larger plaintext than \cite{liu2023lattice}\cite{LWEchannel2022}, i.e., a larger $K$. Let $\mathsf{DEC}((v-\bs^{T}\bu)/\left\lceil {q}/p\right\rfloor): \mathbb{Z}_p^{n}  \mapsto \{0,1\}^K$ be the decoding function. The choices of $\mathsf{ENC}(\cdot)$ and $\mathsf{DEC}(\cdot)$ are discussed in the Subsection \ref{sssec:encoding}. The key generation function of SC-Kyber is the same as Algorithm \ref{alg:kyber_keygen}. Thus we omit details here. The encryption and decryption algorithms of SC-Kyber are described in Algorithms \ref{alg:SCkyber_enc} and \ref{alg:SCkyber_dec}.

%We expect that coding outperforms the compression in terms of reducing CER.
\vspace{-5mm}

\begin{algorithm}[H]
\caption{$\mathsf{SC-Kyber.CPA.Enc}$ $(pk=(\boldsymbol{\mathrm{t}},\rho),m\in\{0,1\}^{K})$}
\label{alg:SCkyber_enc}
\begin{algorithmic}[1]

	\State
	$r \leftarrow \{0,1\}^{256}$
	
	%\State  $\mathrm{\boldsymbol{\mathrm{t}}}\coloneqq\mathsf{Decompress}_{q}(\boldsymbol{\mathrm{t}},d_{t})$
	
	\State
	$\boldsymbol{\mathrm{A}}\sim R_{q}^{k\times k}\coloneqq\mathsf{Sam}(\rho)$
	
	\State  $(\boldsymbol{\mathrm{r}},\boldsymbol{\mathrm{e}_{1}},e_{2})\sim\beta_{\eta_1}^{k}\times\beta_{\eta_2}^{k}\times\beta_{\eta_2}\coloneqq\mathsf{Sam}(r)$
	
	\State  $\boldsymbol{\mathrm{u}}\coloneqq Q_{\mathsf{MMSE},L_u}(\buone)$\label{line:scu}
	
	\State  $v\coloneqq\boldsymbol{\mathrm{t}}^{T}\boldsymbol{\mathrm{r}}+e_2+\left\lceil {q}/p\right\rfloor \cdot \mathsf{ENC}(m) $ \label{line:scv}
	
	\State \Return $c\coloneqq(\mathsf{Index}_{L_u}(\boldsymbol{\mathrm{u}}),v)$

\end{algorithmic}
\end{algorithm}

\vspace{-12mm}

\begin{algorithm}[H]
\caption{${\mathsf{SC-Kyber.CPA.Dec}}\ensuremath{(sk=\bs,c=(\mathsf{Index}_{L_u}(\boldsymbol{\mathrm{u}}),v))}$}
\label{alg:SCkyber_dec}
\begin{algorithmic}[1]

    \State
    $\bu\coloneqq \mathcal{C}_{L_u}(\mathsf{Index}_{L_u}(\boldsymbol{\mathrm{u}}))$

    \State \Return $\mathsf{DEC}((v-\bs^{T}\bu)/\left\lceil {q}/p\right\rfloor)$

\end{algorithmic}
\end{algorithm}

\vspace{-10mm}

\subsection{Information Theoretic Analysis of SC-Kyber}
We study the maximum possible plaintext size $K$, or equivalently, the minimum possible CER for SC-Kyber. Since $v$ remains uncompressed, according to Theorem \ref{the:CLT_noise}, the decoding problem in SC-Kyber can be formulated by
\begin{equation}
y= \left\lceil {q}/{p}\right\rfloor \cdot \mathsf{ENC}(m) +n_e, \label{ne_sc}
\end{equation}
where $n_{e}\leftarrow \mathcal{N}(0,\sigma _{G}^{2}I_{n})$ and the value of $\sigma_{G}$ is given in Theorem \ref{the:CLT_noise}.

It is evident that the SC-Kyber decoding problem is equivalent to the detection problem in an AWGN channel with a $p$-PAM modulation and $n$ independent channel uses. An upper bound on the plaintext size, i.e., $K$, is given below.

\begin{lemma}[Maximum Plaintext Size] \label{lem:Max_pt}For any $\mathsf{ENC}(m):  \{0,1\}^K \mapsto \mathbb{Z}_p^{n}$, the maximum amount of error-free information that can theoretically be decrypted is
\begin{equation}
K \leq n/2\log_2\left(\dfrac{1+\gamma}{1+\gamma/p^2}\right) \triangleq K_{\mathsf{UB}},
\end{equation}
where 
\begin{equation}
\gamma = \dfrac{\left\lceil {q}/{p}\right\rfloor^2\sum_{i=0}^{p-1} (i-(p-1)/2)^2}{p\sigma_G^2}, \label{snr}
\end{equation}
and the value of $\sigma_{G}$ is given in Theorem \ref{the:CLT_noise}.
\end{lemma}

\begin{proof}
We subtract $\left\lceil {q}/{p}\right\rfloor (p-1)/2 $ from both sides of (\ref{ne_sc}), obtaining
\begin{equation}
\hat{y} = \left\lceil {q}/{p}\right\rfloor \hat{m} + n_e, \label{ne_pam}
\end{equation}
where $\hat{y} = y- \left\lceil {q}/{p}\right\rfloor (p-1)/2$ and $\hat{m}= \mathsf{ENC}(m) - (p-1)/2$. Note that shifting a random variable does not change its entropy. Since $n_{e}\leftarrow \mathcal{N}(0,\sigma _{G}^{2}I_{n})$, the model (\ref{ne_pam}) is equivalent to an AWGN channel with a $p$-PAM modulation and $n$ independent channel uses. The maximum achievable information rate (bit per channel use) can be well-approximated by \cite{PAMRate_2021}
\begin{equation}
\max_{\Pr(\hat{m})}\dfrac{I(\hat{y},\hat{m})}{n} \approx 1/2\log_2\left(\dfrac{1+\gamma}{1+\gamma/p^2}\right) \geq 
\dfrac{K}{n},
\end{equation}
where $\gamma$ is given in (\ref{snr}).
\end{proof}

\begin{remark}
Using Lemma \ref{lem:Max_pt}, a lower bound on CER is given by
\begin{equation}
\mathsf{CER} \geq  \dfrac{knd_u+12n}{K_{\mathsf{UB}}} \triangleq \mathsf{CER_{LB}} .
\end{equation}
\end{remark}

In Table \ref{PAM_Rate}, we compute the values of $K_{\mathsf{UB}}$ and $\mathsf{CER_{LB}}$ for different $p$. We observe that with $p=8$, it is theoretically possible to encapsulate two $256$-bit AES keys in a single ciphertext. In the next subsection, we will present a practical encoding scheme to achieve this goal.
\vspace{-5mm}
\begin{table}[th]
\centering
%\renewcommand{\arraystretch}{1.3} % adds row cushion
%\begin{threeparttable}
\caption{SC-KYBER1024: $K_{\mathsf{UB}}$ and $\mathsf{CER_{LB}}$}
\label{PAM_Rate}\centering
\vspace{-3mm} 
\begin{tabular}{|c|c|c|c|c|}
\hline
$p$ & $2$ & $4$ & $8$  & $16$ \\ \hline
$K_{\mathsf{UB}}$ & $255$ & $505$ & $742$ & $935$ \\ \hline
$\mathsf{CER_{LB}}$ & $56.2$ & $28.4$ & $19.3$ & $15.3$ \\ \hline
\end{tabular}
\vspace{-5mm}
\end{table}

\subsection{Encoding, Decoding, and DFR} \label{sssec:encoding}
We consider a binary error correcting code $(N,K,t)$, where $N=n\log_2(p)$ is the codeword length and $t$ is the number of correctable errors \cite{ecclin2004}. Let $c_{\mathsf{B}} \in \{0,1\}^N$ be a binary codeword. The binary encoder $\{0,1\}^K \mapsto \{0,1\}^N$ and the corresponding binay decoder $\{0,1\}^N \mapsto \{0,1\}^K$ are defined by
\begin{align}
c_{\mathsf{B}}&=\mathsf{B}\text{-}\mathsf{Enc}(m)\nonumber \nonumber \\
m&=\mathsf{B}\text{-}\mathsf{Dec}(c_{\mathsf{B}}). \label{ECC_ENC}
\end{align}
Gray code \cite{Proakis} is used to map every $\log_2(p)$ bits in $c_{\mathsf{B}}$ to a (shifted) $p$-PAM symbol. The bit mapper $\{0,1\}^N \mapsto \mathbb{Z}_p^{n}$ and demapper $\mathbb{Z}_p^{n} \mapsto \{0,1\}^N$ are given by
\begin{align}
x&=\mathsf{Gray}(c_{\mathsf{B}})\nonumber \nonumber \\
c_{\mathsf{B}}&=\mathsf{Gray}^{-1}(x). \label{Gray_mp}
\end{align}

\begin{example}
In Table \ref{Gray_Mapping}, we show the bit‐to‐symbol mapping on (shifted) $8$‐PAM. With Gray Code, only one bit changes state from one position to another.
\end{example}
\vspace{-3mm} 
\begin{table}[th]
\centering
%\renewcommand{\arraystretch}{1.3} % adds row cushion
%\begin{threeparttable}
\caption{Bit‐to‐symbol mapping on (shifted) $8$‐PAM}
\label{Gray_Mapping}\centering
\vspace{-3mm} 
\begin{tabular}{|c|c|c|c|c|c|c|c|c|}
\hline
Symbol & $0$ & $1$ & $2$  & $3$ & $4$ & $5$ & $6$  & $7$ \\ \hline
Bits & $000$ & $001$ & $011$ & $010$ & $110$ & $111$ & $101$ & $100$ \\ \hline
\end{tabular}
\vspace{-5mm}
\end{table}

Combining (\ref{ECC_ENC}) and (\ref{Gray_mp}), the encoding and decoding functions in Algorithms \ref{alg:SCkyber_enc} and \ref{alg:SCkyber_dec} are given by
\begin{align}
x&= \mathsf{ENC}(m)=\mathsf{Gray}(\mathsf{B}\text{-}\mathsf{Enc}(m))\nonumber \nonumber \\
m&=\mathsf{DEC}(x)=\mathsf{B}\text{-}\mathsf{Dec}(\mathsf{Gray}^{-1}(\left\lceil x\right\rfloor \bmod p)). \label{enc_dec_sc}
\end{align}
We then derive the DFR for the proposed encoding scheme. 
\begin{lemma}
With $p$-PAM, Gray mapping, and a binary code $(N,K,t)$, for a large $\gamma$ in (\ref{snr}), the DFR of SC-Kyber can be computed by
\begin{equation}
\delta =\sum_{j=t+1}^N
\begin{pmatrix}
N \\j
\end{pmatrix}
\mathsf{RBER}^j(1-\mathsf{RBER})^{N-j}, \label{DFR_SC}
\end{equation}
where $\mathsf{RBER} < 2Q( \left\lceil {q}/(2p)\right\rfloor/\sigma_G)/\log_2(p)$, $Q(\cdot)$ is the Q-function, and the value of $\sigma_{G}$ is given in Theorem \ref{the:CLT_noise}.
\end{lemma}

\begin{proof}
Since SC-Kyber decoding is equivalent to decoding in a $p$-PAM input AWGN channel, the raw symbol error rate (RSER) is 
\begin{align}
\mathsf{RSER} &= 
\begin{cases*}
        {} 2Q( \left\lceil {q}/(2p)\right\rfloor/\sigma_G), & \text{interior points} \\
        {} Q( \left\lceil {q}/(2p)\right\rfloor/\sigma_G), & \text{end points}
    \end{cases*}
\end{align}
Since Gray code makes sure that the most likely symbol errors cause only one bit error, for a large $\gamma$, the raw bit error rate (RBER) is well-approximated by
\begin{equation}
\mathsf{RBER} \approx \mathsf{RSER}/\log_2(p). 
\end{equation}
For a binary code $(N,K,t)$, the DFR can be computed by (\ref{DFR_SC}).
\end{proof}

\begin{example}
We consider an encoding approach with $8$-PAM, Gray mapping, and binary BCH$(768,638,13)$ \cite{ecclin2004}, which is shortened from BCH$(1023,893,13)$. For the SC-KYBER1024, we have $\delta=2^{-174}$ and $\mathsf{CER}=22.5$. Since the original KYBER1024 has $\mathsf{CER}=49$, the communication overhead is reduced by $54\%$. Compared with the bounds in Table \ref{PAM_Rate}, there is a gap of about $0.4$ bit per channel use between the theoretical limit and the rate achieved by the BCH code.  
\end{example}

In Table \ref{com_bch}, we compare the performances of different variants of KYBER1024, given the same security parameters $(n, q, k, \eta_1,\eta_2)$. With the proposed semi-compressed variant, we can easily encrypt two $256$-bit AES keys into a single ciphertext. By changing the code rate, either DFR or CER can be minimized.

\vspace{-5mm} 

\begin{table}[th]
\centering
\caption{Parameters for different variants of KYBER1024: $(n=256, q=3329, k=4, \eta_1=2,\eta_2=2)$}
\label{com_bch}\centering
\vspace{-3mm} 
\begin{threeparttable}[b]
\begin{tabular}{|c|c|c|c|c|c|c|c|c|}
\hline
 Codes & Compression & $p$ & $K$ & $d_u$ & $d_v$ & CER  &  DFR &  Source \\ \hline
Uncoded & Original \cite{NISTpqcdraft2023} & $2$ & $256$ & $11$ & $5$ & $49$ & $2^{-174}$ &  \cite{NISTpqcdraft2023}\\ \hline
Leech lattice  & Original \cite{NISTpqcdraft2023} & $8$ &  $380$ & $11$ & $5$ & $33$ & $2^{-226}$  & \cite{liu2023lattice} \\ \hline 
Q-BCH   & Original \cite{NISTpqcdraft2023} & $5$ & $471$ & $11$ & $5$  & $26.6$ & $2^{-174}$ & \cite{LWEchannel2022}\\ \hline
B-BCH$(768,513,26)$ & Lloyd-Max & $8$ &  $513$ & $10$ & $12$  & $26$ & $2^{-268}$ & This work \\ \hline
B-BCH$(768,638,13)$ & Lloyd-Max & $8$ &  $638$&  $11$ & $12$ &  $22.5$ & $2^{-174}$ & This work \\ \hline
\end{tabular}
\vspace{-1mm}
\end{threeparttable}
\end{table}
%\vspace{-5mm}
% $( k=3,\eta_1=2,\eta_2=2)$

\subsection{Security}
 As shown in \cite{LWEchannel2022}, the encoding has no influence on the distribution of the ciphertext. Moreover, the proposed encoding scheme uses a constant-time BCH decoder in \cite{constantBCH2020}, thus it is resistant to timing attacks. Analysis of potential security issues and mitigation strategies is left for future work.

\section{Conclusion}
In this paper, we prove that the the communication overhead of Kyber can be reduced by more than $50\%$, i.e., encapsulating two $256$-bit AES keys in a single ciphertext. Our solution involves a provably optimal quantization and a constant-time encoding scheme. Closed-form expressions of the DFR is derived. From the information theoretic perspective, we derive the maximum possible plaintext size for Kyber, which tells the minimum possible communication overhead. Our results also imply that it is possible to design a more compact M-LWE encryption scheme than Kyber, using advanced quantization and coding technologies.

\begin{credits}
\begin{comment}
\subsubsection{\ackname} A bold run-in heading in small font size at the end of the paper is
used for general acknowledgments, for example: This study was funded
by X (grant number Y).
\end{comment}

\begin{comment}
\subsubsection{\discintname}
\begin{itemize}
\item Data availability statement: We do not analyse or generate any datasets, because our work proceeds within a theoretical and mathematical approach. One can obtain the relevant materials from the references below.
\item Conflict of interest/Competing interests: The authors have no competing interests to declare that are relevant to the content of this article.
\item Financial interests/Non-financial interests: The authors have no relevant financial or non-financial interests to disclose.
\item Authors' contributions: All authors contributed to the study conception and design. All authors read and approved the final manuscript.
\end{itemize}
\end{comment}

\end{credits}

\bibliographystyle{splncs04}
\bibliography{LIUBIB}

\end{document}